\documentclass[12pt]{revtex4-1}

\usepackage[top=1in,bottom=1in,left=0.6in,right=0.8in]{geometry}
\usepackage{latexsym,amscd,amsfonts}
\usepackage{enumerate}
\usepackage{mathrsfs}
\usepackage{graphicx}
\usepackage{dsfont}
\usepackage{graphicx}
\usepackage{dsfont}
\usepackage{color}
\usepackage{xcolor,rotating,epic,eepic}
\usepackage{amssymb}
\usepackage{amsmath}
\usepackage{amsthm}

\newtheorem{theorem}{Theorem}[section]
\newtheorem{theo}[theorem]{Theorem}
\newtheorem{Pro}[theorem]{Proposition}

\newtheorem{lem}[theorem]{Lemma}
\newtheorem{rem}[theorem]{Remark}

\renewcommand{\epsilon}{\varepsilon}

\newcommand{\Z}{\mathbb Z}
\newcommand{\R}{\mathbb R}
\newcommand{\C}{\mathbb C}

\newcommand{\N}{\mathbb N}
\newcommand{\T}{\mathbb T}

\newcommand{\E}{\mathbb E}
\newcommand{\supp}{\mathrm{supp}\;}
\newcommand{\Hom}{H_{\omega}}
\newcommand{\Ho}{H_{0}(\omega)}

\newcommand{\Hdisc}{H_{\omega}^0}
\newcommand{\SD}{\mathrm{S}_{\mathrm{D}}(\R)\,}
\def\dd{\mathrm d}
\def\ee{\mathrm e}
\def\ii{\mathrm i}
\def\DD{\mathrm D}

\begin{document}

\author{Hakim Boumaza}
\affiliation{LAGA, U.M.R. 7539 C.N.R.S, Institut Galil\'ee, Universit\'e de Paris-Nord, 99 Avenue J.-B. Cl\'ement, 93430 Villetaneuse, France}
\email{boumaza@math.univ-paris13.fr}
\author{Hatem Najar}
\affiliation{D\'epartement de Math\'ematiques, Facult\'e des Sciences de Monastir.\\ Avenue de l'environnement, 5019 Monastir -Tunisie}
\email{hatem.najar@ipeim.rnu.tn}

\thanks{\ Research supported by CMCU project 09/G-15-04.}

\baselineskip=20pt
\renewcommand{\theequation}{\arabic{section}.\arabic{equation}}

\title{Lifshitz tails for matrix-valued Anderson models}

\begin{abstract}
This paper is devoted to the study of Lifshitz tails for a continuous matrix-valued Anderson-type model $\Hom$ acting on $L^2(\R^d)\otimes \C^{D}$, for arbitrary $d\geq 1$ and $D\geq 1$. We prove that the integrated density of states of $\Hom$ has a Lifshitz behavior at the bottom of the spectrum. We obtain a Lifshitz exponent equal to $-d/2$ and this exponent is independent of $D$. It shows that the behaviour of the integrated density of states at the bottom of the spectrum of a quasi-d-dimensional Anderson model is the same as its behaviour for a d-dimensional Anderson model.
\end{abstract}

\maketitle

\section{Introduction}\label{sec_intro}

\subsection{A general model}\label{sec_gen_model} We study the Lifshitz tails behaviour of the integrated density of states (IDS for short) of random Schr\"odinger operators of the form:

\begin{equation}\label{eq_model_Ho}
H_{0}(\omega)=-\Delta_d \otimes I_{\DD} + \sum_{n\in \Z^d} V_{\omega}^{(n)}(x-n),
\end{equation}
\noindent acting on the Hilbert space $L^2(\R^d)\otimes \C^{\DD}$, where $d\geq 1$ and $D\geq 1$ are integers, $\Delta_d$ is the d-dimensional Laplacian and $I_{\DD}$ is the identity matrix of order $D$.

\noindent Let $(\Omega,\mathcal{A},\mathsf{P})$ be a complete probability space and let $\omega \in \Omega$. We assume that, for every $n\in \Z^d$, the functions $x\mapsto V_{\omega}^{(n)}(x)$ take values in the space $\SD$ of real symmetric matrices of order $D$ and that these functions are supported on $[-\frac{1}{2},\frac{1}{2}]^d$ and bounded uniformly on $x$, $n$ and $\omega$. The sequence $(V_{\omega}^{(n)})_{n\in \Z^d}$ is a sequence of independent and identically distributed (\emph{i.i.d.} for short) random variables on $(\Omega,\mathcal{A})$. We finally assume that the sequence $(V_{\omega}^{(n)})_{n\in \Z^d}$ is such that the family of random operators $\{H_{0}(\omega) \}_{\omega \in \Omega}$ is $\Z^d$-ergodic. An operator like (\ref{eq_model_Ho}) is also called a quasi-$d$-dimensional Anderson model.

\noindent The vector space $L^2(\R^d)\otimes \C^{\DD}$ is endowed with the usual scalar product:
$$<f,g>_{L^2(\R^d)\otimes \C^{\DD}}=\int_{\R^d} <f(x),g(x)>_{\C^{\DD}} \dd x = \sum_{i=1}^{\DD} \int_{\R^d} f_i(x)\overline{g_i(x)} \dd x ,$$
where $f=(f_1,\ldots,f_{\DD})$, $g=(g_1,\ldots,g_{\DD})$ and $f_i\in L^2(\R^d)\otimes \C$, $g_i\in L^2(\R^d)\otimes \C$ are the $i$-th components of $f$ and $g$.
\vskip 2mm

\noindent As the functions $x\mapsto V_{\omega}^{(n)}(x)$, for every $\omega \in \Omega$ and every $n\in \Z^d$, take values in $\SD$, the operator $\Ho$ is self-adjoint on the Sobolev space $H^2(\R^d)\otimes \C^{\DD}$, for every $\omega \in \Omega$ (see \cite{RS2}). Thus, its spectrum $\sigma(\Ho)$ is included in $\R$. Moreover, due to the hypothesis of $\Z^d$-ergodicity of the family $\{H_{0}(\omega) \}_{\omega \in \Omega}$, there exists a set $\Sigma_0 \subset \R$ with the property: for $\mathsf{P}$-almost-every $\omega \in \Omega$, $\sigma(\Ho)=\Sigma_0$ (see \cite{CFKS}).
\vskip 3mm

\subsection{Existence of the IDS}\label{sec_exist_ids} We want to study the asymptotic behaviour of the IDS associated to $\Ho$ near the bottom of the almost-sure spectrum of $\Ho$. The IDS of $\Ho$ is the repartition function of energy levels, per unit volume, of $\Ho$. To define it properly, we first need to restrict the operator $\Ho$ to boxes of finite volume. We set, for $L\geq 1$ an integer,
\begin{equation}\label{eq_def_box}
C_{L}=\left[-\frac{2L+1}{2},\frac{2L+1}{2}\right]^d .
\end{equation}
Then, we consider $H_{0,C_{L}}(\omega)$ the restriction of $\Ho$ to the Hilbert space $L^2(C_{L})\otimes \C^{\DD}$ with Dirichlet boundary conditions on the border $\partial C_{L}$ of $C_{L}$. To define the IDS, we now consider, for every $E\in \R$, the following thermodynamical limit:
\begin{equation}\label{eq_def_ids}
N_{0}(E)=\lim_{L\to +\infty} \frac{1}{(2L+1)^d} \# \left\{ \lambda \leq E\ |\ \lambda \in \sigma(H_{0,C_{L}}(\omega)) \right\}.
\end{equation}
We have already proved in \cite{B08} that, for the general model $\Ho$, for every $E\in \R$, the limit (\ref{eq_def_ids}) exists and is $\mathsf{P}$-almost-surely independent of $\omega\in \Omega$ (see \cite[Corollary 1]{B08}). The question of the existence of (\ref{eq_def_ids}) involves two problems to solve. First we had to prove that, for every $E\in \R$ and every $\omega \in \Omega$, the cardinal $\# \{ \lambda \leq E |\ \lambda \in \sigma(H_{0,C_{L}}(\omega)) \}$ is finite. Then we had to prove the existence of the limit when $L$ tends to infinity. Both solutions to these two problems rely strongly on the fact that the semigroup $(e^{-t H_{0,C_{L}}(\omega)})_{t>0}$ has an $L^2$-kernel, which is given through a matrix-valued Feynman-Kac formula (see \cite[Proposition 1]{B08}). Once we obtain that the cardinal $\# \{ \lambda \leq E |\ \lambda \in \sigma(H_{0,C_{L}}(\omega)) \}$ is finite, we prove the convergence, as $L$ tends to infinity, of the sequence of Laplace transforms of the counting measures of the spectral values of $H_{0,C_{L}}(\omega)$ smaller than $E$. We prove the convergence of this sequence by using Birkhoff's ergodic theorem which leads to the existence of a Borel measure $\mathfrak{n}_0$ on $\R$, independent of $\omega$, which is the desired limit. We finally set
\begin{equation}\label{eq_def_ds}
\forall E\in \R,\ N_0(E)=\mathfrak{n}_0((-\infty,E]),
\end{equation}
the distribution function of $\mathfrak{n}_0$. The measure $\mathfrak{n}_0$ is called the density of states of $\Ho$.
\vskip 3mm

\subsection{A particular model}\label{sec_part_model} After this review of existence result of the IDS for the general model $\Ho$, we may consider a particular example of such model for which we will be able to prove precise results on Lifshitz tails of the IDS at the bottom of the spectrum.

\noindent We consider
\begin{equation}\label{eq_model_H}
\Hom= -\Delta_d\otimes I_{\DD} + W(x) + \sum_{n\in \Z^d}  \left(
\begin{smallmatrix}
\omega_{1}^{(n)} V_1(x-n) & & 0\\
 & \ddots &  \\
0 & & \omega_{D}^{(n)} V_D(x-n)\\
\end{smallmatrix}\right),
\end{equation}
acting on $L^2(\R^d)\otimes \C^{\DD}$, where $d\geq 1$ and $D\geq 1$ are integers, $\Delta_d$ and $I_{\DD}$ are as in (\ref{eq_model_Ho}). We set:
\begin{equation}\label{eq_model_H_not}
H=-\Delta_d\otimes I_{\DD} + W\quad \mathrm{and}\quad V_{\omega}(x)=\sum_{n\in \Z^d}  \left(
\begin{smallmatrix}
\omega_{1}^{(n)} V_1(x-n) & & 0\\
 & \ddots &  \\
0 & & \omega_{D}^{(n)} V_D(x-n)\\
\end{smallmatrix}\right).
\end{equation}
For the model (\ref{eq_model_H}), we make the assumptions:
\begin{itemize}
\item[(H1)] $W:\ \R^d \to \SD$ is $\Z^d$-periodic, measurable and bounded.

\item[(H2)] $V_1,\ldots,V_{\DD}$ are nonnegative, bounded, measurable, real-valued functions supported on $[-\frac{1}{2},\frac{1}{2}]^d$. Moreover, we assume that, for every $i\in \{1,\ldots, \DD \}$, there exists  a non-empty cube $\mathcal{C}_{i} \subset [-\frac{1}{2},\frac{1}{2}]^d$ which is not reduced to a single point such that $V_{i}>\mathbf{1}_{\mathcal{C}_{i}}$, where $\mathbf{1}_{\mathcal{C}_{i}}$ is the characteristic function of $\mathcal{C}_{i}$.

\item[(H3)] For every $i\in \{1,\ldots, \DD \}$, $(\omega_{i}^{(n)})_{n\in \Z^d}$ is a family of \emph{i.i.d.} random variables on a complete probability space $(\widetilde{\Omega}_i,\widetilde{\mathcal{A}}_i, \widetilde{\mathsf{P}}_i)$, which are bounded, and whose support of their common law $\nu_i$ contains zero and is not reduced to this single point. Moreover, we assume that,
\begin{equation}\label{eq_cond_va}
\forall i\in \{1,\ldots,\DD \},\ \limsup_{\epsilon \longrightarrow 0^+} \frac{\log |\log \widetilde{\mathsf{P}}_i(\omega_{i}^{(0)} \leq \epsilon)|}{\log \epsilon}=0.
\end{equation}

\end{itemize}
\vskip 2mm
\noindent In particular, we can take Bernoulli random variables for the $\omega_{i}^{(n)}$'s. By adding a suitable constant diagonal matrix to the periodic background $W$, we may always assume that the $\omega_{i}^{(n)}$'s are nonnegative valued (because of their boundedness). If we set
$$(\Omega,\mathcal{A},\mathsf{P})=\left(\bigotimes_{n\in \Z^d}(\widetilde{\Omega}_1\otimes \cdots \otimes \widetilde{\Omega}_{\DD}),\bigotimes_{n\in \Z^d}(\widetilde{\mathcal{A}}_1 \otimes \cdots \otimes \widetilde{\mathcal{A}}_{\DD}), \bigotimes_{n\in \Z^d}(\widetilde{\mathsf{P}}_1 \otimes \cdots \otimes \widetilde{\mathsf{P}}_{\DD})\right),$$
then $(\Omega,\mathcal{A},\mathsf{P})$ is a complete probability space and $\{\Hom\}_{\omega \in \Omega}$ is $\Z^d$-ergodic because of the non-overlapping of the random variables $\omega_{i}^{(n)}$. We denote by $\Sigma$ the almost-sure spectrum of $\Hom$. By adding a suitable scalar matrix $\lambda I_{\DD}$ to the periodic potential $W$, we may always assume that $\mathbf{\inf \Sigma =0}$, where $\Sigma$ is the almost sure spectrum of the $\Z^d$-ergodic family $\{\Hom\}_{\omega \in \Omega}$.

\noindent The model (\ref{eq_model_H}) is a particular case of (\ref{eq_model_Ho}) for which the potential split into a deterministic periodic part $W$ and a random part $V_{\omega}$ which appears as a diagonal matrix. We will denote by $N:\ E\to N(E)$ the IDS of $H_{\omega}$.

\begin{rem}
If we assume that, at least for one $x\in \R^d$, $W(x)$ is not a diagonal matrix, then we cannot write $\Hom$ as a direct sum $\otimes_{i=1}^{\DD} H_{\omega,i}$ of scalar-valued operators $H_{\omega,i}$ acting on $L^2(\R^d)\otimes \C$ and for which all the results we will present here are already known.
\end{rem}

\begin{rem}
The hypothesis $(H2)$ on the boundedness of the functions $V_i$ and the boundedness of their support implies in particular that each $V_i$ is in $L^p(\R^d)\otimes \C^{\DD}$ with $p=2$ if $d\leq 3$, $p>2$ if $d=4$ and $p>\frac{d}{2}$ if $d\geq 5$. These are the assumptions made in \cite{K99}.
\end{rem}
\vskip 3mm

\noindent For $d=1$, matrix-valued operators as (\ref{eq_model_H}) are also called quasi-one-dimensional Anderson models. Localization results in both dynamical and spectral senses for such models, for particular simple choices of $W$ and $V_1,\ldots,V_{\DD}$, are obtained in \cite{B09} and \cite{B10}. These quasi-one-dimensional models are of physical interest as they can be considered as partially discrete approximations of Anderson models on a two-dimensional continuous strip. Such a two-dimensional Anderson model on a continuous strip can, by example, modelize electronic transport in nanotubes. Indeed, a two dimensional continuous Anderson model is defined by:
\begin{equation}\label{eq_def_Hcs}
H_{cs}(\omega)=-\Delta_2 + \sum_{n\in \Z} \omega^{(n)} V(x-n,y),
\end{equation}

\noindent acting on $L^2(\R \times [0,1])\otimes \C$ with Dirichlet boundary conditions on $\R\times \{0\}$ and $\R\times \{1\}$. The $\omega^{(n)}$'s are \emph{i.i.d.} random variables and $V$ is supported in $[0,1]^2$. As the continuous strip $\R\times [0,1]$ has one finite length dimension ($[0,1]$) and one infinite length dimension ($\R$), we can physically consider this strip as a quasi-one-dimensional nanotube. The spectral properties of $H_{cs}(\omega)$ describe properties of the electronic transport in the nanotube $\R \times [0,1]$.

\subsection{The behaviour of the IDS}\label{sec_heuristic_ids} The main result of this paper is about Lifshitz tails for the IDS $N(E)$ of $\Hom$ at the bottom of the spectrum. In 1963, Lifshitz (see \cite{L63}) had conjecture that, for a continuous random Schr\"odinger operator of IDS $N(E)$, there exist $c_{1},c_{2}>0$ such that $N(E)$ satisfies the asymptotic:
\begin{equation}\label{eq_Lif_heuristic}
N(E)\simeq c_1\exp({-c_2(E-E_{0})^{-\frac{d}{2}}}),
\end{equation}
as $E$ tends to $E_0$, where $E_0$ is the bottom of the spectrum of the considered Schr\"odinger operator. The behaviour (\ref{eq_Lif_heuristic}) is known as Lifshitz tails (for more details, see part IV.9.A of \cite{PF}) and the exponent $-d/2$ is called the Lifshitz exponent of the operator. The principal results known on Lifshitz tails are mainly shown for Schr\"odinger operators, in both continuous and discrete cases (see \cite{K89,KM83,K99,N77,P77,S85} and others) and for Schr\"odinger operators with magnetic fields (see \cite{K10,N03} ). Up to our knowledge, all studied examples of Schr\"odinger operators are for scalar-valued operators, and no adaptation of the known results to matrix-valued operators like (\ref{eq_model_H}) has been done yet.
\vskip0.5cm

\noindent In a previous article of one of the author (see \cite{B08}), we had already obtain a result of H\"older continuity of the IDS for a particular example of model $\Hom$, in dimension $d=1$. For $d=1$, we can use the formalism of transfer matrices and define Lyapunov exponents for $\Hom$ and, in this case, the sum of the positive Lyapunov exponents is harmonically conjugated to the IDS of $\Hom$, through a so-called Thouless formula (see \cite[Theorem 3]{B08}). It allows us to prove that under some assumptions on the group generated by the transfer matrices, the F\"urstenberg group of $\Hom$, we have positivity of the Lyapunov exponents, their H\"older continuity with respect to the energy parameter and thus, by the Thouless formula, the same H\"older regularity for the IDS. This assumptions are hard to verify for a general $D$ (where $D$ is the size of the matrix-valued potential), but we were able to verify them for a very particular example of $\Hom$ in dimension $d=1$, where the periodic and random potentials $W$ and $V_{\omega}$ acts like constant functions. In \cite{B10}, we studied the following Anderson operator:
\begin{equation}\label{eq_model_Hl}
H_{\ell}(\omega) = -\frac{\dd^2}{\dd x^2}\otimes I_{\DD} + V + \sum_{n\in \Z}  \left(
\begin{smallmatrix}
c_1 \omega_{1}^{(n)} \mathbf{1}_{[0,\ell]}(x-\ell n) & & 0\\
 & \ddots &  \\
0 & & c_{\DD} \omega_{\DD}^{(n)} \mathbf{1}_{[0,\ell]}(x-\ell n)\\
\end{smallmatrix}\right),
\end{equation}
\noindent acting on $L^2(\R)\otimes \C^{\DD}$, where $D\geq 1$ is an integer, $I_{\DD}$ is the identity matrix of order $D$ and $\ell>0$ is a real number. The matrix $V$ is a real $\DD\times \DD$ symmetric matrix. The constants $c_1,\ldots,c_{\DD}$ are non-zero real numbers. For $I \subset \R$, $\mathbf{1}_{I}$ is the characteristic function of $I$. The random variables $\omega_i^{(n)}$ are like in model (\ref{eq_model_H}). As we can see, $H_{\ell}(\omega)$ is a particular example of $\Hom$ with $W$ constant and $V_i=\mathbf{1}_{[0,\ell]}$, for every $i\in \{1,\ldots,\DD\}$.
\noindent For this operator we had obtained the following regularity result:
\begin{Pro}\cite[Proposition 6.2]{B10}\label{prop_ids_holder}
For Lebesgue-almost every $V\in \SD$, there exist a finite set $\mathcal{S}_V\subset \R$ and a real number $\ell_C:=\ell_C(\DD,V) >0$ such that, for every $\ell\in (0,\ell_C)$, there exists a compact interval $I(\DD,V,\ell)\subset \R$ such that, if $I\subset I(\DD,V,\ell)\setminus \mathcal{S}_V$ is an open interval, then the integrated density of states of $H_{\ell}(\omega)$, $E\mapsto N_{\ell}(E)$, is H\"older continuous on $I$.
\end{Pro}

\begin{rem}
We actually proved even more : in such an open interval $I$ with $\Sigma \cap I \neq \emptyset$, we have Anderson localization in both spectral and dynamical senses.
\end{rem}

\noindent The Proposition \ref{prop_ids_holder} is interesting in itself but doesn't give any information about the behaviour of the IDS at the bottom of spectrum and, until now, it was not clearly stated that it has a Lifshitz behaviour. One of the motivation of the present article is to fill this lack of information of the IDS for quasi-one-dimensional operators and in particular those like $H_{\ell}(\omega)$ we have studied before from the localization point of view.

\subsection{The result}\label{sec_result} We can now state the main result of the present article.
\begin{theo}\label{thm_main}
Let $\Hom$ be the operator defined by (\ref{eq_model_H}) and let $N$ be its IDS. We assume hypothesis $(H1)$, $(H2)$  and $(H3)$ and we also assume that  $\inf \Sigma = 0$. Then,
\begin{equation}\label{eq_main_thm}
\lim_{E \to 0^{+}} \frac{\log|\log \big( N(E)-N(0^+) \big)|}{\log (E)}= -\frac{d}{2}.
\end{equation}
In particular, this limit does not depend on D.
\end{theo}

\begin{rem}
\noindent (1) Under some assumption on the behavior of the integrated density of states of the background operator $H$, it might be possible to obtain a result for internal bands.

\noindent (2) Theorem \ref{thm_main} could be used to give a different proof of localization than the one provided in \cite{B09} (see \cite{DS01,NTh,Nloca,V02}).
\end{rem}

\noindent It is important to insist on the fact that the Lifshitz exponent $-d/2$ obtained here does not depend on the integer $D\geq 1$. It means that, looking only at the Lifshitz behaviour of the IDS at the bottom of the spectrum of the considered operator, we cannot distinguished a quasi-$d$-dimensional Anderson model like (\ref{eq_model_H}) from a $d$-dimensional Anderson model (for $D=1$).

\noindent One of the motivations in considering matrix-valued Anderson models is that we could expect, as $D$ tends to infinity, that we could obtain information about a $(d+1)$-dimensional Anderson model from a quasi-$d$-dimensional Anderson model. In particular, by obtaining a localization result for (\ref{eq_model_Hl}) for an arbitrary $D\geq 1$, we could have expect to obtain a similar localization result for the continuous strip (\ref{eq_def_Hcs}). The presence of the Lifshitz tails behaviour of the IDS is usually a strong sign of the presence of localization at the bottom of the spectrum. If we wanted to use the Lifshitz tails behaviour of the IDS to prove localization (like in \cite{DS01}) and at the same time following the idea of approaching a $(d+1)$-dimensional Anderson model by a quasi-$d$-dimensional Anderson model, we would have expected a Lifshitz exponent depending on $D$ in a way such that this exponent would tend to $-(d+1)/2$ as $D$ tends to infinity. But, Theorem \ref{thm_main} contradicts this. So, we actually obtained an argument in favor of the idea that we cannot really get a proof of localization in dimension $2$ (or more generally in dimension $d+1$) by an approximation procedure using quasi-one-dimensional Anderson models. This, at least if we follow a localization proof based upon the Lifshitz tails behaviour of the IDS.

\section{Matrix-valued Floquet theory}\label{sec_floquet}

\subsection{Matrix-valued Floquet decomposition} In this section, we review the main results of the Floquet theory for the deterministic operator
\begin{equation}\label{eq_H_per}
H=-\Delta_d\otimes I_{\DD} + W
\end{equation}
acting on $L^2(\R^d)\otimes \C^{\DD}$, with periodic potential $W$, and we adapt them to the matrix-valued setting.
\noindent More precisely, we assume here that $W$ is a $\Z^d$-periodic function in $L^p(\R^d)\otimes \SD$ with $p=2$ if $d\leq 3$, $p>2$ if $d=4$ and $p>\frac{d}{2}$ if $d\geq 5$. If $W$ is $\Z^d$-periodic, measurable and bounded as in model (\ref{eq_model_H}), it is in such an $L^p(\R^d)\otimes \SD$ space. Then, $H$ is essentially selfadjoint on $C^{\infty}_0(\R^d)\otimes \C^{\DD}$ (the space of compactly supported function, $\C^{\DD}$-valued, of class $C^{\infty}$) with domain the Sobolev space $H^2(\R^d)\otimes \C^{\DD}$ \cite{RS2}.

\noindent First of all, let us notice that the formalism and all the general results about constant fiber direct integrals are still valid in our setting of matrix-valued operators. We refer to \cite[Section XIII.16]{RS4} for a complete presentation of these results.

\noindent For $y\in \R^d$, we denote by $\tau_y$ the operator of translation by $y$ which is defined, for $u\in L^2(\R^d)\otimes \C^{\DD}$ and $x\in \R^d$, by $(\tau_y u)(x)=u(x-y)$. Then, because $W$ is $\Z^d$-periodic, the operator $H$ is invariant by conjugation by $\tau_n$, for every $n\in \Z^d$: $\forall n\in \Z^d,\ \tau_n \circ H \circ \tau_n^* = \tau_n \circ H \circ \tau_{-n} = H.$

\noindent Thus, $H$ is a $\Z^d$-periodic operator. Let $(e_1,\ldots,e_d)$ be the canonical basis of $\R^d$. We recall that $C_0$ can be considered as the fundamental cell of the lattice $\Z^d$,
$$C_0=\left\{ x_1 e_1 + \ldots + x_d e_d\ \Big|\ \forall j\in \{1,\ldots,d\}, -\frac{1}{2} \leq x_j \leq \frac{1}{2} \right\}.$$
If $C_0^*$ is the fundamental cell of the dual lattice $(\Z^d)^* \simeq 2\pi \Z^d$, then $C_0^*$ is identified to the torus $\T^*=\R^d /2\pi\Z^d$. Let $\theta \in \T^{*}$. We denote by $\mathcal{D}_{\theta}^{'}$ the space of $\C^{\DD}$-valued, $\theta$-quasiperiodic distributions in $\R^d$, which is the space of distributions $u\in \mathcal{D}^{'}(\R^d)\otimes \C^{\DD}$ such that, for any $n\in \Z^d$, $\tau_n u= \ee^{-\ii n.\theta}u$. Let $\mathcal{H}_{\theta}=(L_{\mathrm{loc}}^2 (\R^d)\otimes \C^{\DD})\cap \mathcal{D}_{\theta}^{'}$, endowed with the norm on $L^2(C_0)\otimes \C^{\DD}$.
\noindent We also define, for $k\in \Z$, the spaces $\mathcal{H}_{\theta}^{k}=(H_{\mathrm{loc}}^k (\R^d)\otimes \C^{\DD})\cap \mathcal{D}_{\theta}^{'}$, where $H_{\mathrm{loc}}^k (\R^d)\otimes \C^{\DD}$ is the space of distributions that locally belong to the Sobolev space $H^k(\R^d)\otimes \C^{\DD}$.
\noindent In order to define the Fourier decomposition we will use later, it remains to define the space:
$$\mathcal{H}=\big\{ u\in (L_{\mathrm{loc}}^2 (\R^d) \otimes L^2(\T^*))\otimes \C^{\DD}\ \big| \ \forall (x,\theta,n)\in \R^d \times \T^* \times \Z^d,\ u(x+n,\theta)=\ee^{\ii n.\theta}u(x,\theta) \big\},$$
endowed with the norm:
$$\forall u\in \mathcal{H},\ ||u||_{\mathcal{H}}= \frac{1}{\mathrm{vol}(\T^*)}\int_{\T^*} ||u(.,\theta)||_{L^2(C_0)\otimes \C^{\DD}} \; \dd \theta.$$
For $\theta \in \R^d$ and $u\in L^2(\R^d)\otimes \C^{\DD}$, we define $Uu\in \mathcal{H}$ by
\begin{equation}\label{eq_def_U}
\forall x\in \R^d,\ \forall \theta \in \T^*,\ (Uu)(x,\theta)=\sum_{n\in \Z^d} \ee^{\ii n.\theta}(\tau_n u)(x)= \sum_{n\in \Z^d} \ee^{\ii n.\theta} u(x-n).
\end{equation}
Actually, the expression (\ref{eq_def_U}) is well-defined for $u\in \mathcal{S}(\R^d)\otimes\C^{\DD}$, the Schwartz space, and by Parseval theorem, this expression can be extended as an isometry from $L^2(\R^d)\otimes\C^{\DD}$ to $\mathcal{H}$.
\noindent For every $v\in \mathcal{H}$, we can define $U^*$, the inverse of $U$ by:
\begin{equation}\label{eq_def_Ustar}
\forall x\in \R^d,\ (U^{*}v)(x)=\frac{1}{\mathrm{vol}(\T^*)}\int_{\T^*} v(x,\theta)\; \dd \theta.
\end{equation}
Indeed, we have, for $v\in \mathcal{H}$ and $x\in \R^d$,
\begin{eqnarray}
(UU^{*}v)(x) & = & \sum_{n\in \Z^d} \ee^{\ii n.\theta} (U^{*}v)(x-n) =  \sum_{n\in \Z^d} \ee^{\ii n.\theta} \frac{1}{\mathrm{vol}(\T^*)}\int_{\T^*} v(x-n,\theta)\; \dd \theta \nonumber \\
 & = & \sum_{n\in \Z^d} \ee^{\ii n.\theta} \frac{1}{\mathrm{vol}(\T^*)}\int_{\T^*} \ee^{-\ii n.\theta}v(x,\theta)\; \dd \theta  = \sum_{n\in \Z^d} \ee^{\ii n.\theta} \hat{v}_n(x) = v(x,\theta). \nonumber
\end{eqnarray}
Thus, $U^{*}$ is a left inverse for $U$ which is an isometry from $L^2(\R^d)\otimes\C^{\DD}$ to $\mathcal{H}$, therefore $U$ is unitary and $U^{*}=U^{-1}$.
\noindent To obtain a Floquet decomposition for the operator $H$, it remains to prove that the operators $U$ and $H$ commute. As, for any $j\in \{1,\ldots,d\}$ and any $n\in \Z^d$, the partial derivation $\partial_j$ commute with the translation $\tau_n$, we have $[\partial_j,U]=0$. Thus, $[-\Delta_d \otimes I_{\DD},U]=0$. Then, using the $\Z^d$-periodicity of $W$, we also have, for every $u\in L^2(\R^d)\otimes\C^{\DD}$, every $x\in \R^d$ and every $\theta \in \T^*$,
\begin{eqnarray}
(U\circ W)(u)(x,\theta) & = &  \sum_{n\in \Z^d} \ee^{\ii n.\theta} W(x-n)u(x-n) =  \sum_{n\in \Z^d} \ee^{\ii n.\theta} W(x)u(x-n) \nonumber \\
 & = & W(x) \sum_{n\in \Z^d} \ee^{\ii n.\theta} u(x-n) = W(x) (Uu)(x,\theta) = (W\circ U)(u)(x,\theta), \nonumber
\end{eqnarray}
as at $x$ fixed, the multiplication by $W(x)\in \SD$ is continuous. Thus, $[W,U]=0$ and we finally have $[H,U]=0$. As we can see, even in the matrix-valued case we still have that $H$ and $U$ commute.
 Following \cite{RS4}, we deduce that $H$ admits the Floquet decomposition:
\begin{equation}\label{eq_Floquet}
UHU^{*}=\int_{\T^*}^{\oplus} H_{\theta} \; \dd \theta,
\end{equation}
where $H_{\theta}$ is the selfadjoint operator $H$ acting on $\mathcal{H}_{\theta}$ with domain $\mathcal{H}_{\theta}^2$. Having this Floquet decomposition, we can continue to follow \cite{RS4} to obtain that $H_{\theta}$ has a compact resolvent. It is a consequence on the assumptions made on the $L^p$-regularity of $W$ which ensure that $H$ is elliptic. As $H_{\theta}$ has a compact resolvent, its spectrum is discrete and we denote by
$$E_0(\theta)\leq E_1(\theta)\leq \ldots \leq E_j(\theta) \leq \ldots$$
its eigenvalues, called the \emph{Floquet eigenvalues} of $H$. Moreover, the functions $\theta \mapsto E_j(\theta)$, for $j\in \N$, are continuous and, if $j$ tends to infinity, then $E_j(\theta)$ tends to $+\infty$, uniformly in $\theta$. Actually, as $H_{\theta}$ depends analytically on $\theta$, we also have that $\theta \mapsto E_j(\theta)$ is an analytic function in the neighborhood of any point $\theta^{0} \in \T^*$ such that $E_j(\theta^0)$ is an eigenvalue of multiplicity one of $H_{\theta^0}$.

\noindent The set $E_j(\T^*)$ is a closed interval called the $j$-th spectral band of $H$ and the spectrum of $H$ is given by
$$\sigma(H)=\bigcup_{j\in \N} E_j(\T^*).$$
If $d\geq 2$, the bands can overlap, but it is not the case in dimension $1$ (except maybe at an edge point).

\begin{rem}
Heuristically, the difference between the usual scalar-valued case ($D=1$) and the matrix-valued case is that they are ``$D$ times'' more Floquet eigenvalues in the matrix-valued case and thus the multiplicities of the $E_j(\theta)$'s are \emph{a priori} bigger in the matrix-valued case than in the scalar-valued case.
\end{rem}

\noindent We finish this section by proving a result of nondegeneracy of the minimum of the first Floquet eigenvalue.

\begin{Pro}\label{pro2.2}
Let $\theta^0 \in \T^*$ be a minimum of $\theta \mapsto E_0(\theta)$. Then, there exist $\delta >0$ and $C>0$ such that
\begin{equation}\label{eq_min}
\forall \theta \in \T^*,\ |\theta-\theta^0|<\delta \ \Rightarrow
\ C |\theta-\theta^0|^2 \leq E_0(\theta)-E_0(\theta^0) \leq |\theta-\theta^0|^2 .
\end{equation}
\end{Pro}

\noindent This proposition means that a minimum of the function $\theta \mapsto E_0(\theta)$ is always nondegenerate.

\begin{proof}
We will follow the ideas of \cite{KS87} and \cite{K99} and adapt them to the matrix-valued case. Let $\mathcal{S}(\cdot,\theta^0)\in \mathrm{GL}_{\mathrm{D}}(\C)$ be the fundamental solution of the differential system $H_{\theta^0} \psi = E_0(\theta^0) \psi$. We have,
\begin{equation}\label{eq_min1}
H_{\theta^0} \mathcal{S}(\cdot,\theta^0) = E_0(\theta^0) \mathcal{S}(\cdot,\theta^0).
\end{equation}
We define a new scalar product on the space $\mathcal{H}_{\theta^0}$ by:
$$\forall f,g \in \mathcal{H}_{\theta^0},\ <f,g>_{\theta^0} = \int_{C_0} <\mathcal{S}(x,\theta^0)f(x), \mathcal{S}(x,\theta^0) g(x)>_{\C^{\mathrm{D}}} \dd x.$$
We denote by $\tilde{\mathcal{H}}_{\theta^0}$ the Hilbert space $(\mathcal{H}_{\theta^0},<\cdot,\cdot>_{\theta^0} )$. Then, for every $\theta \in \T^*$, we define on $\tilde{\mathcal{H}}_{\theta^0}$ the operator $\tilde{H}_{\theta}$ of domain
$$D(\tilde{H}_{\theta})=\left\{ u\in \mathcal{H}_{\theta^0} \ |\ \mathcal{S}(\cdot,\theta^0)u \in \mathcal{H}_{\theta}^2 \right\},$$
and given by
$$\forall u \in D(\tilde{H}_{\theta}),\ \forall x\in \R^d,\ (\tilde{H}_{\theta}u)(x)=\mathcal{S}(x,\theta^0)^{-1} (H_{\theta} - E_0 (\theta^0)) (\mathcal{S}(x,\theta^0)u)(x).$$
Let  $\theta \in \T^*$. For every $u \in D(\tilde{H}_{\theta})$,
\begin{eqnarray}
 <u,\tilde{H}_{\theta} u>_{\theta^0} & = &   \int_{C_0} <  (\mathcal{S}(x,\theta^0)u)(x),(H_{\theta} - E_0 (\theta^0)) (\mathcal{S}(x,\theta^0)u)(x) >_{\C^{\mathrm{D}}} \dd x \nonumber \\
 & = & \sum_{i=1}^D  \int_{C_0} ((\mathcal{S}(x,\theta^0)u)(x))_i \ \overline{((H_{\theta} - E_0 (\theta^0)) (\mathcal{S}(x,\theta^0)u)(x))_i} \dd x , \nonumber
\end{eqnarray}
where $(v)_i$ denote the i-th coordinate of the vector $v\in \C^{\mathrm{D}}$. Let $i\in \{1,\ldots,D\}$. Since $\mathcal{S}(x,\theta^0)$ is invertible for every $x\in C_0$, one can perform the same change of coordinates in
$$ \int_{C_0} ((\mathcal{S}(x,\theta^0)u)(x))_i \ \overline{((H_{\theta} - E_0 (\theta^0)) (\mathcal{S}(x,\theta^0)u)(x))_i} \dd x$$
than the one done in \cite[(4.6)]{DS84}. Then we can follow the proof of \cite[Proposition 4.4D]{DS84} to obtain :
$$\int_{C_0} ((\mathcal{S}(x,\theta^0)u)(x))_i \ \overline{((H_{\theta} - E_0 (\theta^0)) (\mathcal{S}(x,\theta^0)u)(x))_i} \dd x = \int_{C_0} ((\mathcal{S}(x,\theta^0) \nabla u)(x))_i \ \overline{((\mathcal{S}(x,\theta^0) \nabla u)(x))_i } \dd x$$
and thus, by summing over  $i\in \{1,\ldots,D\}$,
\begin{equation} \label{eq_min2}
  \forall u \in D(\tilde{H}_{\theta}),\ <u,\tilde{H}_{\theta} u>_{\theta^0}  =  || \nabla u ||_{\theta^0}^2.
\end{equation}
We have, using (\ref{eq_min2}) at the second equality, for every $u \in D(\tilde{H}_{\theta})$,
\begin{eqnarray}
 ( E_0(\theta)-E_0(\theta^0))<u,u>_{\theta^0} & = & <u,E_0(\theta)u>_{\theta^0} - <u,\mathcal{S}(x,\theta^0)^{-1} H_{\theta} \mathcal{S}(x,\theta^0) u>_{\theta^0} \nonumber \\
 &\ & + <u,\mathcal{S}(x,\theta^0)^{-1} (H_{\theta}-E_0(\theta^0)) \mathcal{S}(x,\theta^0) u>_{\theta^0} \nonumber \\
 & = & -  <u,\mathcal{S}(x,\theta^0)^{-1} (H_{\theta}- E_0(\theta)) \mathcal{S}(x,\theta^0) u>_{\theta^0} +  || \nabla u ||_{\theta^0}^2.
\end{eqnarray}
Since, by definition, $E_0(\theta)$ is the smallest eigenvalue of the compact resolvent operator $H_{\theta}$, one has
{\small $$ <u,\mathcal{S}(x,\theta^0)^{-1} (H_{\theta}- E_0(\theta)) \mathcal{S}(x,\theta^0) u>_{\theta^0} =  \int_{C_0} <  (\mathcal{S}(x,\theta^0)u)(x),(H_{\theta} - E_0 (\theta)) (\mathcal{S}(x,\theta^0)u)(x) >_{\C^{\mathrm{D}}} \dd x \geq 0.$$ }
Thus,
\begin{equation}
\forall u \in D(\tilde{H}_{\theta}),\ ( E_0(\theta)-E_0(\theta^0)) || u ||_{\theta^0}^2 \leq || \nabla u ||_{\theta^0}^2.
\end{equation}
Moreover, if one choose $u_0=S(\cdot,\theta^0)^{-1}S(\cdot,\theta)(1,0,\ldots,0)$, then $S(\cdot,\theta^0)u_0$ is an eigenfunction of $H_{\theta}$ associated to $E_0(\theta)$ and $ <u_0,\mathcal{S}(\cdot,\theta^0)^{-1} (H_{\theta}-E_0(\theta)) \mathcal{S}(\cdot,\theta^0) u_0>_{\theta^0} = 0.$
For such a particular $u_0$, $( E_0(\theta)-E_0(\theta^0)) || u_0 ||_{\theta^0}^2 = || \nabla u_0 ||_{\theta^0}^2$ and we have
\begin{equation}\label{eq_min3}
 E_0(\theta)-E_0(\theta^0) = \inf \left\{ \frac{||\nabla u ||_{\theta^0}^2}{||u||_{\theta^0}^2} \ ;\ u \in D(\tilde{H}_{\theta}) \right\}.
\end{equation}
Now we choose an other particular $u_1$ in $\mathcal{H}_{\theta^0}$ :  $u_1=(0,\ldots,0,\ee^{\ii (\theta -\theta^0)\cdot},0\ldots,0)$. Then,
$$||\nabla u_1 ||_{\theta^0}^2 = ||(\theta-\theta^0)(0,\ldots,0,\ee^{\ii (\theta -\theta^0)\cdot},0\ldots,0)||_{\theta^0}^2 = |\theta - \theta^0 |^2 ||u_1||_{\theta^0} .$$
Setting this $u_1$ in (\ref{eq_min3}), one gets
\begin{equation}\label{eq_min4}
E_0(\theta)-E_0(\theta^0) \leq \frac{||\nabla u_1 ||_{\theta^0}^2}{||u_1 ||_{\theta^0}^2} = |\theta - \theta^0|^2.
\end{equation}
To prove the lower bound in (\ref{eq_min}), let $u_2 \in D(\tilde{H}_{\theta})$ be such that
\begin{equation}\label{eq_min5}
  \frac{||\nabla u_2 ||_{\theta^0}^2}{||u_2||_{\theta^0}^2} =  E_0(\theta)-E_0(\theta^0).
\end{equation}
The vector $u_2$ is a vector where $u\mapsto <u, \tilde{H}_{\theta} u >_{\theta^0}$ is minimal. Since $\{ H_{\theta} \}_{\theta \in \T^*}$ is an analytic family of operators, one can use its Taylor expansion near $\theta^0$ to write
$$\exists \delta >0,\ \exists C>0,\ \forall \theta \in \T^*, |\theta - \theta^0| < \delta,\  \int_{C_0} <  (\mathcal{S}(x,\theta^0)u_2)(x),(H_{\theta} - E_0 (\theta^0)) (\mathcal{S}(x,\theta^0)u_2)(x) >_{\C^{\mathrm{D}}} \dd x $$
\begin{eqnarray}
\ & \geq & \int_{C_0} <  (\mathcal{S}(x,\theta^0)u_2)(x),(H_{\theta^0} - E_0 (\theta^0)) (\mathcal{S}(x,\theta^0)u_2)(x) >_{\C^{\mathrm{D}}}\dd x \label{eq_min6} \\
 &\ & + (\theta-\theta^0) \sum_{j=1}^d \frac{\partial}{\partial \theta_j} \left.  \int_{C_0} <  (\mathcal{S}(x,\theta^0)u_2)(x),H_{\theta} (\mathcal{S}(x,\theta^0)u_2)(x) >_{\C^{\mathrm{D}}} \dd x \right|_{\theta=\theta^0} \label{eq_min7} \\
 &\ & + C|\theta-\theta^0 |^2 \int_{C_0}<(\mathcal{S}(x,\theta^0)u_2)(x),(\mathcal{S}(x,\theta^0)u_2)(x) >_{\C^{\mathrm{D}}} \dd x \label{eq_min8} \\
 & \geq & C |\theta-\theta^0 |^2 || u_2 ||_{\theta^0}^2.
\end{eqnarray}
Indeed, the integral in (\ref{eq_min6}) is equal to zero by definition of $\mathcal{S}(\cdot,\theta^0)$. Moreover, the derivatives in (\ref{eq_min7}) vanish since $\theta^0$ is where $\theta \mapsto E_0(\theta)$ is minimal and $u_2$ verifies (\ref{eq_min5}). Finally, for the inequality and the existence of $C>0$ in (\ref{eq_min8}), one uses a compactness argument as in the proof of \cite[Proposition 1.1]{K99}. So we have
\begin{equation}\label{eq_min9}
 \exists \delta >0,\ \exists C>0,\ \forall \theta \in \T^*, |\theta - \theta^0| < \delta,\ \frac{<u_2,\tilde{H}_{\theta} u_2>_{\theta^0}}{||u_2||_{\theta^0}^2} \geq C  |\theta-\theta^0 |^2.
\end{equation}
Using (\ref{eq_min2}), (\ref{eq_min4}), (\ref{eq_min5}) and (\ref{eq_min9}), one finally gets :
{\small $$ \exists \delta >0,\ \exists C>0,\ \forall \theta \in \T^*,|\theta - \theta^0| < \delta,\ \ C|\theta-\theta^0 |^2 \leq \frac{<u_2,\tilde{H}_{\theta} u_2>_{\theta^0}}{||u_2||_{\theta^0}^2} =  \frac{||\nabla u_2 ||_{\theta^0}^2}{||u_2||_{\theta^0}^2} =  E_0(\theta)-E_0(\theta^0) \leq |\theta-\theta^0 |^2 ,$$ }
which achieves the proof.
\end{proof}

\subsection{Wannier basis} We recall concepts used in  \cite{K99,N03}. Let ${{\mathcal{E}}}\subset L^2(\R^d)\otimes \C^{\DD}$ be a closed subspace, invariant by $\Z^d$-translations, \emph{i.e.}, for every $n\in \Z^d,\ \Pi^{{\mathcal{E}}}=\tau_n^{*}\Pi^{{\mathcal{E}}}\tau_n$, where $\Pi^{{\mathcal{E}}}$ is the orthogonal projection on $\mathcal{E}$.

\noindent As $\Pi^{{\mathcal{E}}}$ is $\Z^d$-periodic, it admits a Floquet decomposition similar to the one of $H$ and, using the orthogonality, one gets:
$$\Pi^{{\mathcal{E}}}=\int_{\T^*}^{\oplus} \Pi_{\theta}^{{\mathcal{E}}} \; \dd \theta,$$
where $\Pi_{\theta}^{{\mathcal{E}}}$ is the operator $\Pi^{{\mathcal{E}}}$ acting on $\mathcal{H}_{\theta}$. The operator $\Pi_{\theta}^{{\mathcal{E}}}$ is therefore an orthogonal projection acting on $L^2(C_0)\otimes \C^{\DD}$. As for $(H_{\theta})_{\theta \in \T^*}$, the family $(\Pi_{\theta}^{{\mathcal{E}}})_{\theta \in \T^*}$ is continuous in $\theta$ and thus is of constant rank. If we fix $\theta \in \T^*$, we can find an orthonormal system $(w_{m,0})_{m\in M}$, with $M\subset \N$ a set of indices independent of $\theta$, that spans the range of $\Pi_{\theta}^{{\mathcal{E}}}$. Taking the image by $U^*$ of this orthonormal system, one gets an orthonormal system $(\widetilde{w}_{m,0})_{m\in M}$. If we set, for $n\in \Z^d$, $\widetilde{w}_{m,n}=\tau_n(\widetilde{w}_{m,0})$, then $(\widetilde{w}_{m,n})_{(m,n)\in M \times \Z^d}$ is an orthonormal basis of $\mathcal{E}$. Such a system is called a \emph{Wannier basis} of $\mathcal{E}$. The vectors $(\widetilde{w}_{m,0})_{m\in M}$ are called the \emph{Wannier generators} of $\mathcal{E}$.
\vskip 2mm

\noindent Let $\mathcal{E} \subset L^2(\R^d)\otimes \C^{\DD}$ be a space which is invariant by $\Z^d$-translations. The closed subspace $\mathcal{E}$ is said to be \emph{of finite energy for $H$} if $\Pi^{\mathcal{E}}H \Pi^{\mathcal{E}}$ is a bounded operator. In this case, $\mathcal{E}$ admits a finite set of Wannier generators. We now assume that $\mathcal{E}$ is of finite energy for $H$.
\vskip 2mm

\noindent Let $J_{0}$ be the set of indices of the Floquet eigenvalues of $H$ which take the value $0$ for some values of $\theta \in \T^*$. We identify $J_{0}$ to $\{1,...,n_{0}\}$. Let $Z$ be the set of
$\theta \in \T^*$ for which there exists $j\in J_{0}$ such that, $E_j(\theta)=0$. When $\theta^0$ is a nondegenerate minimum of $E_{j}$, $Z$ is a set of isolated points (see \cite{K99}). It occurs when the density of states $\mathfrak{n}$ has a nondegenerate behavior at $0$ (see \cite{K78}). For $j \in J_{0}$, we define $Z_j=\{\theta\in \T^*\ ;\ E_j(\theta)=0\}$. The sequence $(Z_j)_{j\in J_{0}}$ is decreasing for the inclusion and $Z_1=Z$. For $\theta^0\in Z$, $M_{\theta ^0} \subset \N$ is the set of indices such that $E_j(\theta^0)=0$.

\noindent We will denote by $w_j(\cdot,\theta)$ a Floquet eigenvector associated with the Floquet eigenvalue $E_j(\theta)$ of $H$. For $(\theta, \theta^{\prime})\in (\T^*)^2$, we define $T_{\theta \to \theta^{\prime}}:\mathcal{H}_{\theta}\to {\mathcal{H}}_{\theta^{\prime}}$ by:
$$\forall v\in \mathcal{H}_{\theta},\ \forall x\in \R^d,\ (T_{\theta \to \theta^{\prime}}v)(x)=\ee^{\ii x\cdot(\theta^{\prime}-\theta)}v(x).$$

\begin{lem}\label{lemMM2} There exists $(v_j(\cdot, \theta))_{j\in J_{0}}$, a family of functions on ${\mathcal{H}}_{\theta}$, such that:

\noindent 1) for $\theta^0\in Z$ and $j\in M_{\theta^0}$, there exists $V_{\theta^0}$ a neighborhood of $\theta^0$ in $\T^*$ such that the map $\theta \in V_{\theta^0} \mapsto v_j(\cdot, \theta)\in \mathcal{H}_{\theta}$ is real analytic (i.e. $\theta \mapsto T_{\theta \to \theta^0}v_j(\cdot, \theta)$ is analytic as a function from $V_{\theta ^0}$ to $\mathcal{H}_{\theta^{0}}$) and, for $\theta \in V_{\theta^0}$, $\mathrm{span}\langle (v_j(\cdot, \theta))_{j\in M_{\theta^0}}\rangle = \mathrm{span}\langle (w_j(\cdot, \theta))_{j\in M_{\theta^0}}\rangle$.

\noindent 2) For $\theta \in \T^*$, the system $(v_j(\cdot, \theta))_{j\in M_{\theta^0}}$ is orthonormal in $\mathcal{H}_{\theta}$ and $\mathrm{span}\langle (w_{j}(\cdot,\theta))_{j\in J_{0}}\rangle=\mathrm{span}\langle (v_{j}(\cdot,\theta))_{j\in J_{0}}\rangle$.
\end{lem}
\begin{proof}
We refer to \cite{K99,N03}, Lemma $3.1$.
\end{proof}
\vskip 3mm

\noindent In the next section, we will use this notion of Wannier basis and the notations we have just introduce to reduce our problem on estimating $N(E)-N(0^+)$ to a discrete problem.

\section{Reduction of the problem}\label{sec_loc_energy}
The goal of this section is to give an estimate of $N(E)-N(0^+)$ for an energy $E$ close to $0$. This will be accomplished by means of the IDS of certain reference operators, which are discrete operators. In this section we will use the notations introduced in Section \ref{sec_floquet}.

\subsection{Reduction to a discrete problem}\label{sec_disc_pb}
The reduction procedure consists into decomposing the operator $\Hom$ according to various translation-invariant subspaces. The random operators thus obtained are what we consider as reference operators. They will be used to prove the upper bound on the IDS.

\noindent We denote by $\Pi_{0}(\theta)$ the orthogonal projection in $\mathcal{H}_{\theta}$ on the vector space generated by $(w_j(\cdot,\theta))_{j\in J_{0}}$.
\noindent One defines
\begin{equation}\label{eq_def_pi_alpha}
\Pi_{0} = U^{-1}\Big( \int_{\T^*}\Pi_{0}(\theta) \dd \theta \Big) U\ \ :\ \ \ L^2(\R^d)\otimes \C^{\DD} \to L^2(\R^d)\otimes \C^{\DD}.
\end{equation}
$\Pi_{0}$ is an orthogonal projection on $L^2(\R^d)\otimes \C^{\DD}$ and, for every $n \in \Z^d$, we have $\tau_{n}^{*} \Pi_{0}\tau_{n}=\Pi_{0}$. Thus, $\Pi_{0}$ is $\Z^d$-periodic. We set $\mathcal{E}_{0} = \Pi_{0}(L^2(\R^d)\otimes\C^{\DD})$. This space is translation-invariant because of the $\Z^d$-periodicity of $\Pi_{0}$. Moreover ${\mathcal{E}}_{0}$ is of finite energies for $H$ as defined in (\ref{eq_model_H_not}). The main result justifying this reduction procedure is the following theorem which compares $E\mapsto N(E)$, the IDS of $\Hom$, to $E\mapsto N_{{\mathcal{E}}_{0}}$, the IDS of the discretize operator $\Hdisc = \Pi_{0} \Hom \Pi_{0}$.

\begin{theo}\label{thm_bb22}
Let $\Hom$ be defined by (\ref{eq_model_H}) with the assumptions $(H1)$, $(H2)$ and $(H3)$. There exist $\varepsilon>0$ and $C>1$ such that, for $0\leq E\leq \varepsilon$ we have
\begin{equation}\label{eq_thm_bb22}
0\leq N(E)-N(0^+)\leq N_{{\mathcal{E}_{0}}}(C\cdot E),
\end{equation}
where $N_{\mathcal{E}_{0}}$ is the IDS of the discretized operator $\Hdisc =\Pi_{0} \Hom \Pi_{0}$.
\end{theo}

\begin{proof}
See Theorem 4.1 in \cite{K99}.
\end{proof}

\subsection{Periodic approximations}\label{sec_per_approx}
In order to get bounds on the density of states of $\Pi_0 \Hom \Pi_0$, we will now define periodic approximations of the operator $\Hom$. For these approximations, we will be able to control the density of state near $0$ by comparing it to some reduced operators. Then by taking a limit on the density of state of the reduced operators, we can get bounds on the density of states of $\Pi_0 \Hom \Pi_0$ and thus on the density of states of $\Hom$ itself by using Theorem \ref{thm_bb22}.
Let $k$ an integer larger than $1$. We define the following periodic operator
\begin{equation}\label{eq_per_approx}
H_{\omega,k}=-\Delta_d \otimes I_{\DD} + W(x) + \sum_{n\in C_k \cap \Z^d}\  \sum_{\beta \in (2k+1)\Z} \left(\begin{smallmatrix}
\omega_1^{(n)} V_{1} (x-(n+\beta))&  & 0\\
 & \ddots &  \\
0 &  & \omega_{\DD}^{(n)}V_{\DD} (x-(n+\beta))
\end{smallmatrix}
\right).
\end{equation}

\noindent The operator $H_{\omega, k}$ is $(2k+1){\Z^d}$-periodic and essentially selfadjoint. It is an $H$-bound perturbation of $H$ with relative bound zero. Because of the $(2k+1){\Z^d}$-periodicity, we introduce the torus $\T_{k}^*=\R^d/(2(2k+1)\pi\Z^d)$. We also define $N_{\omega,k}$, the IDS of $H_{\omega,k}$ by
\begin{equation}\label{eq_def_ids_per_approx}
N_{\omega,k}(E)=\frac{1}{(2\pi)^d}\sum_{j\in \N}\int_{\{\theta \in \T_k^*,\ E_{\omega,k,j}(\theta)\leq E \}} \dd \theta.
\end{equation}
where $E_{\omega,k,j}$ is the $j$-th Floquet eigenvalue of the periodic operator $H_{\omega,k}$. Let $dN_{\omega, k}$ be the derivative of $N_{\omega,k}$, in the distribution sense. As $E\mapsto N_{\omega, k}(E)$ is an increasing function, $dN_{\omega,k}$ is a positive measure, it is the density of states of $H_{\omega,k}$. Then, by \cite{K99,RS4}, for every $\varphi\in C_{0}^{\infty}(\R)$, the distribution $dN_{\omega,k}$ verifies
\begin{equation}\label{eq_1sla}
\langle \varphi,dN_{\omega,k} \rangle  =  \frac{1}{(2\pi)^d}\int_{\theta \in\T_k^{*}} \mathrm{tr}_{\mathcal{H}_{\theta}} \big(\varphi(H_{\omega,k,\theta}) \big)\dd \theta, =  \frac{1}{\mathrm{vol}(C_{k})}\mathrm{tr}\Big(\mathbf{1}_{C_{k}} \varphi(H_{\omega,k})\mathbf{1}_{C_{k}} \Big),
\end{equation}
where $\mathrm{tr}(A)$ is the trace of a trace-class operator $A$. We index this trace by $\mathcal{H}_{\theta}$ if the trace is taken in $\mathcal{H}_{\theta}$ and here, the operator $\mathbf{1}_{C_{k}} \varphi(H_{\omega,k})$ is a trace-class operator. The proof of (\ref{eq_1sla}) is given in \cite[Proposition 5.1]{K99}.
\vskip 1mm
\noindent We want to take a limit on the density of states $dN_{\omega,k}$ of the periodic approximations in order to recover properties of the density of states of $\Hom$ from properties of $dN_{\omega,k}$. The following theorem ensure that it is possible.
\begin{theo}\label{thm_lim_ids_per_approx}
1) For any $\varphi\in C_{0}^{\infty}(\R)$ and for almost every $\omega\in \Omega$, we have
$$\lim_{k\rightarrow \infty}\langle \varphi,dN_{\omega,k}\rangle=\langle \varphi,dN\rangle.$$
2) For any $E \in \R$ a continuity point for $N$, we have $\displaystyle\lim_{k\rightarrow \infty} \E(N_{\omega,k}(E))=N(E)$.
\end{theo}
\begin{proof}
The result of Theorem \ref{thm_lim_ids_per_approx} is close to that of Theorem 5.1 of \cite{K99}. The proof is also similar and is based on functional analysis.
\end{proof}

\section{Proof of theorem \ref{thm_main}}\label{sec_main_thm}

\noindent We will proceed in two steps. First, we will prove a lower bound and then an upper bound.

\subsection{Lower bound}\label{sec_low_bd}
In this subsection we prove

\begin{theo}\label{LL1}
Let $H_{\omega}$,  be the operator defined by (\ref{eq_model_H}) with the assumptions $(H1)$, $(H2)$ and $(H3)$. We have
\begin{equation}\label{hifr}
\liminf_{E\longrightarrow 0^{+}} \frac{\log\Big|\log \Big ( N(E)-N(0^+)\Big)\Big|}{\log E}\geq - \frac{d}{2}.
\end{equation}
\end{theo}
\vskip 3mm

\noindent \textrm{\textbf{{The proof of Theorem \ref{LL1}:}}} As $0$ is the bottom of the spectrum, for $\varepsilon >0$ we have $N(\varepsilon) - N(0)=N(\varepsilon) - N(-\varepsilon)$.
 To prove Theorem \ref{LL1}, we will lower bound $N(\varepsilon)-N(-\varepsilon)$. Then, for $L$ large, we will show that $H_{\omega,C_L}$ (we recall that $H_{\omega,C_L}$ is $H_{\omega}$
  restricted to $C_L$ with Dirichlet boundary conditions) has a large number of eigenvalues in $[-\varepsilon,\varepsilon]$ with a large probability. To do this we will construct a family
  of approximate eigenvectors associated to approximate eigenvalues of $H_{\omega, C_L}$ in $[-\varepsilon, \varepsilon]$. These functions will be constructed from an eigenvector of
  $-\Delta_d\otimes I_{\DD} + W(x)$ associated to $0$. Locating this eigenvector in $\theta $ and imposing to $\omega_1^{(n)}$ to be small for $n$ in some well chosen box, one obtains an
  approximate eigenfunction
of $H_{\omega, C_L}$. Locating the eigenfunction in $x$ in several disjointed places, we get several eigenfunctions two by two orthogonal.

\noindent In order to simplify the notations, we assume in what follows that $\theta^0=0$ is a point where $E_0(\theta)$ reaches $0$. From the same arguments as in \cite{N03} and using
Proposition \ref{pro2.2}, there exists $C>0$ such that, for $\tilde{f}(\cdot, \theta):=(f_1(\cdot,\theta),\cdots ,f_D(\cdot,\theta))=v_1(\cdot, \theta)$ in $L^2(\R^d)\otimes \C^{\DD}$, ($v_{1}$ is the vector constructed in Lemma \ref{lemMM2}) one has
\begin{equation}\label{42}
||(-\Delta_d\otimes I_{\DD} + W(x))\tilde{f}(\cdot,\theta)||_{L^2(C_{0})\otimes \C^{\DD}} \leq C |\theta|^2.
\end{equation}
This is due to the fact that locally near $\theta=0$, we can
reduce the study of $\displaystyle -\Delta_d\otimes I_D+W(x)$
which is analytic in $\theta$ and is equal to $0$ when $\theta=0$
and to the use of (\ref{eq_min}) and (\ref{eq_min1}).\newline
 We
assume, without loss of generality, that $f_1\neq 0$ and we set
\begin{equation}\label{mai2}
f(\cdot,\theta)=\frac{f_1(\cdot,\theta)}{|\theta_1|}(1,0,\cdots,0).
\end{equation}
Let $0<\xi<1$ be a small constant. Let $\chi \in C_{0}^{\infty}(\R)$ being positive, supported in $[\frac{\xi}{2}, \xi]$ and such that $\displaystyle\int_{[\frac{\xi}{2},\xi]}\chi(t)^2dt=2$.

\noindent For $\varepsilon > 0$, we define
\begin{equation}\label{mai3}
{\mathcal{W}}_{\varepsilon}(\theta)=\varepsilon
^{-d/4}\prod_{j=1}^{d}\chi(\varepsilon^{-\frac{1}{2}}\theta_j) \in L^2(\T^*)\quad \mbox{and}\quad {\mathcal{W}}_{\varepsilon}^{f}(\cdot,\theta)={\mathcal{W}}_{\varepsilon}(\theta)\cdot f(\cdot,\theta)\in L^2(C_{0})\otimes \C^{\DD}.
\end{equation}
Now let us estimate
$\displaystyle ||(-\Delta_d\otimes I_{\DD} + W(x)){\mathcal{W}}_{\varepsilon}^f||^2_{\mathcal{H}}$. We have
$$
||(-\Delta_d\otimes I_{\DD} + W(x)){\mathcal{W}}_{\varepsilon}^f||^2_{\mathcal{H}}=\frac{1}{\mbox{vol}(\T^*)} \int_{\T^*}||(-\Delta_d\otimes I_{\DD} + W(x))(\theta)f(\cdot,\theta)||^2_{L^2(C_{0})\otimes \mathbb{C}^D} | \mathcal{W}_{\varepsilon}(\theta)|^2 \dd \theta .
$$
So using (\ref{42}) and (\ref{mai3}), we get
\begin{equation}\label{*}
||(-\Delta_d\otimes I_{\DD} + W(x)){\mathcal{W}}_{\varepsilon}^f||_{{\mathcal{H}}}^2 \leq C^2 \int_{\T^*}|\theta|^4|\ \mathcal{W}_{\varepsilon}(\theta)|^2 \dd \theta \leq  C^2\varepsilon^2\int_{[\frac{\xi}{2},\xi]^d}|\theta|^4\prod_{j=1}^{d}\chi^2(\theta_j) \dd \theta \leq \frac{\varepsilon^2}{8},
\end{equation}
if $\xi$ is small enough. For $\beta \in \Z^d$, we define
$$
\mathcal{W}_{\varepsilon,\beta}^{f}(\cdot,\theta)=e^{-i\beta \cdot \theta}\mathcal{W}_{\varepsilon}^{f}(\cdot,\theta)\quad \mbox{and} \quad \mathcal{W}_{\alpha,\varepsilon,\beta,\zeta}^{f}(\cdot,\theta)=e^{-i\beta \cdot
\theta}(\Pi_{\Lambda_{\alpha}(\zeta)}{\mathcal{W}}_{\varepsilon}^{f})(\cdot,\theta),
$$
where $\Lambda_{\alpha}(\zeta)$ is the cube defined by
$$
\Lambda_{\alpha}(\zeta)=\left\{n \in {\Z^d}\ \Big| \ \mbox{for}\ 1\leq j\leq d, \ |n _{j}|\leq \zeta^{-(\frac{1}{2}+\alpha)}\right\}
$$
and $\Pi_{\Lambda_{\alpha}(\zeta)}$ is the orthogonal projection on
$\Lambda_{\alpha}(\zeta)$, i.e it is the operator of orthogonal
projection on $L^2(\mathbb{T}^*)$ on the space spanned by vectors
$\displaystyle \theta \to e^{i\gamma \theta},\gamma \in
\Lambda_\alpha (\zeta)$.

\noindent We set
$$
\mathcal{U}_{\varepsilon,\beta}^f(x)=\int_{\T^*}\mathcal{W}_{\varepsilon,\beta}^{f}(x,\theta) \dd \theta \quad \mbox{and}\quad \mathcal{U}_{\alpha,\varepsilon,\beta,\zeta}^f(x)=\int_{\T^*} \mathcal{W}_{\alpha,\varepsilon,\beta,\zeta}^{f}(x,\theta) \dd\theta.
$$
For $L$ large and $\beta$ and $(\omega_1^{(n)})_{n\in \Z^d}$ well chosen, $\mathcal{U}_{\alpha,\varepsilon,\beta,\zeta}^{f}$ will be an approximate eigenfunction of $H_{\omega,C_L}$ associated to an approximate eigenvalue in the interval ${[-\varepsilon,\varepsilon]}.$

\noindent We notice that $\displaystyle {\mathcal{U}}_{\alpha,\varepsilon,\beta,\zeta}^{f}\in L^2(\R^d)\otimes\C^{\DD}$ and $\displaystyle \mathcal{U}_{\alpha,\varepsilon,\beta,\zeta}^{f_1}\in {L^2(\R^d)}$. As in \cite{N03} one gets that
$$\big|\big|{\mathcal{U}}_{\alpha,\varepsilon,\beta,\zeta}^{f} \big|\big|_{L^2(\R^d)\otimes\C^{\DD}} \geq \big|\big|{\mathcal{U}}_{\alpha,\varepsilon,\beta,\zeta}^{f_1} \big|\big|_{L^2(\R^d)} > C > 0.$$
Now we have to look to the conditions under which we have
\begin{equation}\label{eq_low_bound_maj}
 \Big|\Big|\Big(-\Delta_d\otimes I_{\DD} + W(x)\Big){\mathcal{U}}_{\alpha,\varepsilon,\beta,\zeta}^f\Big|\Big|^2_{L^2(\R^d)\otimes \C^{\DD}}\leq \varepsilon^2.
\end{equation}

Note that
\begin{eqnarray}
\Big|\Big|H_{\omega,C_{L}} \mathcal{U}_{\alpha,\varepsilon,\beta,\zeta}^f\Big|\Big|_{L^2(\R^d)\otimes \C^{\DD}}^2 & \leq & \Big|\Big| H_\omega \cdot \mathcal{U}_{\alpha,\varepsilon,\beta,\zeta}^f\Big|\Big|_{L^2(\R^d)\otimes \C^{\DD}}^2 \nonumber \\
& \leq & 2\Big|\Big| \Big(-\Delta_d\otimes I_{\DD} + W(x)\Big){\mathcal{U}}_{\alpha,\varepsilon,\beta,\zeta}^f\Big|\Big|_{L^2(\R^d)\otimes \C^{\DD}}^2 + 2\Big|\Big|V_\omega \mathcal{U}_{\alpha,\varepsilon,\beta,\zeta}^f\Big|\Big|_{L^2(\R^d)\otimes \C^{\DD}}^2 .\label{koum}
\end{eqnarray}
Equation (\ref{eq_low_bound_maj}) give the bound on the first member of (\ref{koum}). It just remains to control the second term. To do so, one needs the following lemma

\begin{lem}\label{sed1}
Let $\zeta=\varepsilon$. There exists $K>0$, such that
\begin{equation}\label{49}
\Big|\Big|V_{\omega}\cdot
\mathcal{U}_{\alpha,\varepsilon,\beta,\varepsilon}^f
\Big|\Big|_{L^2(\R^d)\otimes \C^{\DD}}^2 \leq \varepsilon^4 + K
\cdot \Big( \sup_{n \in \beta +2\Lambda_{\alpha}(\varepsilon
)}\omega_1^{(n)} \Big)^2.
\end{equation}
\end{lem}

\noindent Before proving this lemma let us use it to finish the proof of Theorem \ref{LL1}.
\vskip 2mm

\noindent Taking (\ref{eq_low_bound_maj}) and (\ref{koum}) into account, we get that there exists $K>0$ such that
\begin{equation}
\Big|\Big|H_{\omega}\cdot \mathcal{U}_{\alpha,\varepsilon,\beta,
\varepsilon }^f\Big|\Big|^2 \leq 4\varepsilon^2+K\Big( \sup_{n \in
\beta
+2\Lambda_{\alpha}(\varepsilon)}\omega_1^{(n)}\Big)^2.\label{53}
\end{equation}
Now, for $L$ large, we may divide $C_L$ into $L(\varepsilon)$ disjoints cubes of size $2\Lambda_{\alpha}(\varepsilon)$. For $\alpha<\frac{1}{2}$, there exists $C>0$ such that $L(\varepsilon)$ satisfies
\begin{equation}
 L(\varepsilon)\simeq\frac{(2L)^d}{\varepsilon^{-d(\frac{1}{2}+\alpha)}} \geq\frac{(L\varepsilon )^d}{C} \label{53b}.
\end{equation}
We can find $\beta_1,\ldots, \beta_{L(\varepsilon)}$ in $\Z^d$ such that :
$$\bigcup _{j=1}^{L(\varepsilon)} (\beta_j+ 2\Lambda_{\alpha}(\varepsilon)) \subset C_L \quad \mbox{and, for}\ j\neq j',\quad (\beta _j+2\Lambda_{\alpha}(\varepsilon  ))\cap (\beta _{j'}+2\Lambda_{\alpha}(\varepsilon))=\emptyset.$$
In particular, for  $j\neq j',$ $\mathcal{U}_{\alpha,\varepsilon,\beta_j,\varepsilon}^f$ and $\mathcal{U}_{\alpha,\varepsilon, \beta_{j'}, \varepsilon}^f$ are orthogonal. Then,
\begin{eqnarray}\label{54}
\E\Big(\# \Big\{\mbox{eigenvalues of}\ \Pi _{C_L} H_{\omega}\Pi_{C_L} \mbox{ in } [-\varepsilon,\varepsilon]\Big\}\Big) &
\geq & \E \Big(\#\Big\{j\in \{1,\ldots, L(\varepsilon)\}\ \Big|\ ||H_\omega \mathcal{U}_{\alpha,\varepsilon,\beta_{j},\varepsilon}^f||_{L^2(\R^d)\otimes \C^{\DD}}\leq \varepsilon\Big\}\Big) \nonumber \\
&  \geq & \E \left(\sum_{j=1}^{L(\varepsilon)}B_j(\omega) \right),
\end{eqnarray}
where
$$
B_j (\omega)=\left\lbrace \begin{array}{ccl}
 1\quad & \mbox{ if } K \cdot \left(\sup_{n \in \beta_j +2\Lambda_{\alpha}(\varepsilon)} \omega_{1}^{(n)}\right)^2 \leq \frac{\varepsilon^2}{2}.\\[3mm]
 0\quad & \mbox{ if not. }
\end{array}\right.$$
The $(B_j)_{1\leq j \leq L(\varepsilon)}$ are \emph{i.i.d.} Bernoulli random variables. So equations (\ref{54}) and (\ref{53b}) imply that there exists $C>0$ such that one has
\begin{equation*}
\frac{1}{(2L+1)^d}\E \Big(\# \Big\{ \mbox{ eigenvalues of }\Pi_{C_L}H_{\omega}\Pi_{C_L} \mbox{ in } [-\varepsilon, \varepsilon]\Big\}\Big) \geq  \frac{L(\varepsilon)}{(2L+1)^d} \mathsf{P}(B_1=1) \geq \frac{1}{C} \varepsilon^d \mathsf{P}(B_1=1).
\end{equation*}
Hence, taking the limit $L\to \infty$, we get that, for $\varepsilon >0$ small,
\begin{equation}\label{55}
N(\varepsilon)- N(-\varepsilon) \geq \frac{1}{C} \varepsilon^{d} \mathsf{P}(B_1=1).
\end{equation}
It just remains to estimate $\mathsf{P}(B_1=1)$. If, for $1\leq
j\leq L(\varepsilon)$ and $n\in
\beta_j+2\Lambda_{\alpha}(\varepsilon)$, one has $\displaystyle
\omega_{1}^{(n)} \leq \frac{\varepsilon\cdot \varepsilon }{2 K}$,
then for $\varepsilon$ rather small
$$K \Big( \sup_{n \in \beta_j +2\Lambda_{\alpha}(\varepsilon )}\omega_{1}^{(n)} \Big)^2\leq \frac {\varepsilon^2}{2}.$$
As the random variables are \emph{i.i.d.}, one has the estimate
 $$\mathsf{P}(B_1=1)\geq \tilde{\mathsf{P}}_1\Big(\omega_{1}^{(0)}\leq \frac{\varepsilon}{2 K }\Big)^{{\# (2\Lambda_{\alpha}(\varepsilon)})}.$$
 Hence, taking the double logarithm of (\ref{55}), using assumption (H3) and the fact that $\#(2\Lambda_{\alpha}(\varepsilon))\simeq \varepsilon^{-(\frac{d}{2}+d\cdot\alpha)} $, we get that
\begin{equation}
\lim_{\varepsilon \to 0^{+}} \frac{\log\Big|\log \Big ( N(\varepsilon)-N(0)\Big)\Big|}{\log\varepsilon}\geq-
\frac{d}{2}-d\alpha.\label{sed3}
\end{equation}
The equation (\ref{sed3}) is  true for any $\alpha > 0$, by letting $\alpha$ tend to $0$, we end the proof of Theorem \ref{LL1}. $\Box$

\noindent It remains to prove Lemma \ref{sed1} to finish this section on the lower bound.
\vskip 3mm

\noindent {\bf{The proof of  Lemma \ref{sed1}.}} We have :
\begin{equation}
\Big|\Big|V_{\omega}\cdot\mathcal{U}_{\alpha,\varepsilon,\beta,
\varepsilon }^f\Big|\Big|^2_{L^2(C_0)\otimes \C^{\DD}} \lesssim
\Big|\Big|\sum_{n \in \Z^d}\omega_1^{(n)} V_1(x-n)
\mathcal{U}_{\alpha,\varepsilon,\beta,\varepsilon
}^{f_1}\Big|\Big|^2_{L^2(C_0)} . \label{naj}
\end{equation}
Then,
\begin{equation}\label{DIMDIM}
\Big|\Big|\sum_{n \in \Z^d}\omega_1^{(n)}
V_1(x-n)\mathcal{U}_{\alpha,\varepsilon,\beta, \varepsilon
}^{f_1}\Big|\Big|^2 \lesssim \varepsilon^5+ \int_{\R^d}
\Big(\sum_{n\in
\Z^d}\omega_1^{(n)}V_1(x-n)\Big)^2\Big|\mathcal{U}_{\varepsilon,\beta}^{f_1}(x)\Big|^2
\dd x.
\end{equation}
Here we used the fact that $\mathcal{U}_{\alpha,\varepsilon,\beta,
\varepsilon }^{f_1}$ and $\mathcal{U}_{\varepsilon,\beta}^{f_1}$ are
close to each other.  \newline We set
\begin{equation} \label{X**}
S_{\beta,\varepsilon} \leq K \sum_{\eta \in \Z^d} \left(\sum_{n \in
\eta +\Lambda_\alpha(\varepsilon)} \omega_1^{(n)} \right)^2\cdot
\int_{C_{0}}\Big|\mathcal{U}_{\varepsilon,\eta-\beta}^{f_1}(x)\Big|^2
\dd x.
\end{equation}So for our choice of $f_1$ and
using the fact that $V_1$ is supported in $C_0$, we deal with a with
a simple quantity to control using the non-stationary phase and
which was already estimated in \cite{K99, N03}.$\Box$

\noindent To finish the proof of Theorem \ref{thm_main}, it remains to prove the upper bound.

\subsection{Upper bound}\label{sec_up_bd}
We start this section by recalling that, as we deal with the bottom of the spectrum, we have non-degeneracy of the first Floquet eigenvalue at the bottom of the spectrum as shown in Proposition \ref{pro2.2}. Using this, we prove the following theorem:
\begin{theo}\label{LL12}
Let $H_{\omega}$ be the operator defined by (\ref{eq_model_H}) with the assumptions (H1), (H2) and (H3). Then
$$\limsup_{E \longrightarrow 0^{+}}\frac{\log | \log(N(E)-N(0^+))|}{\log E} \leq -\frac{d}{2}.$$
\end{theo}

\noindent \textbf{The proof of Theorem \ref{LL12}.} To prove the upper bound, it is enough to prove the same upper bound on $N_{\mathcal{E}_0}$ (as defined in Theorem \ref{thm_bb22}). To do this, we show that $N_{{\mathcal{E}}_{0}}$ (and so $N$) may be compared to the IDS of some well chosen discrete Anderson model whose behavior of its IDS is already known.

\noindent We begin by isolating the contributions from the various points for which $E_{j}(\theta)$ take the value $0$. We recall that the band at $0$ is generated by $(E_j(\theta))_{1\leq j\leq n_{0}}$. For $1\leq j \leq n_{0} $, $Z_j = \{\theta \in \T^* ; E_j(\theta)=0 \}$. The sequence $(Z_j)_{1\leq j \leq n_{0}}$ is decreasing ($Z_{j+1} \subset Z_j$). Let $\theta^0 \in Z$. We set $j(\theta ^0)=\sup M_{\theta ^0}$ with $M_{\theta^0}=\{j\ ;\ 1\leq j\leq n_{0}, E_{j}(\theta^0)=0\}$. We replace the Floquet eigenvectors $(w_j(\cdot, \theta))_{1 \leq j \leq j(\theta^0)}$ associated to $(E_j(\theta))_{1\leq j\leq j(\theta^0)}$ by the vectors $(v_j(\cdot, \theta))_{1 \leq j \leq j(\theta^0)}$ constructed in Lemma \ref{lemMM2}. They are analytic in a neighborhood $V_{\theta^0}$ of $\theta^0$. Let $\theta$ be close to $\theta^0$. The operator $H^0(\theta)=\Pi_{0} H(\theta) \Pi_{0}$ is unitarily equivalent to the multiplication operator by a function on $L^2(\T^*)$ with values in $\mathcal{M}_{n_0}(\C)$. This matrix-valued function takes the following block diagonal form :
$$\left(\begin{matrix}
B_{j(\theta ^0)}(\theta)& 0 & 0 & ... & 0 \\
0 & E_{j(\theta ^0)+1}(\theta) & 0 & ... & 0 \\
0 & 0 & \ddots & \ddots & : \\
0 & 0 & ... & 0 & E_{n_{0}}(\theta)
\end{matrix}\right),$$
\vskip 2mm

\noindent where the matrix $B_{j(\theta ^0)}(\theta)$ is of size $j(\theta ^0)\times j(\theta ^0)$ and is given by
$$\left(\begin{matrix}
\langle v_{1}(\cdot,\theta),H(\theta
)v_{1}(\cdot,\theta)\rangle_{L^2(C_{0})\otimes\C^{\DD}} & ...& \langle v_{1}(\cdot,\theta),H(\theta )v_{j(\theta^0)}(\cdot,\theta)\rangle_{L^2(C_{0})\otimes \C^{\DD}} \\
: & \ddots & : \\
\langle v_{j(\theta ^0)}(\cdot,\theta),H(\theta)v_{1}(\cdot,\theta)\rangle_{L^2(C_{0})\otimes \C^{\DD}} & ... & \langle v_{j(\theta ^0)}(\cdot,\theta),H(\theta)v_{j(\theta^0)}(\cdot,\theta)\rangle_{L^2(C_{0})\otimes \C^{\DD}}
\end{matrix}\right) .$$
\vskip 2mm

\noindent The matrix $B_{j(\theta ^0)}(\theta)$ has $(E_j(\theta))_{1\leq j\leq j(\theta^0)}$ for eigenvalues. The operator  $V_\omega^0=\Pi_0 V_\omega \Pi_0$ is unitarily equivalent to the multiplication operator by the matrix with entries $(\langle V_\omega v_i,v_j\rangle )_{1\leq i,j\leq n_0}$.

\noindent For $u\in L^2(\T^*)\otimes \C^{\DD},\ \Pi_0 u=\sum_{i=1}^{n_0}\langle u, v_i\rangle_{L^2(C_0)\otimes \C^{\DD}}v_i.$ For $\theta ^0 \in Z$, we set
$$\varpi_{\theta ^0}(\theta)= \sum_{j=1}^d \Big(1-\cos(\theta_j -\theta ^0  _j)\Big).$$
We recall that the eigenvalues $(E_j(\theta))_{1\leq j \leq j(\theta^0)}$ are non-degenerate at $0$.  So there exists $\widetilde{V}_{\theta ^0} $ (an open neighborhood of $\theta ^0$) and $C>1$ such that, for
$\theta \in \widetilde{V}_{\theta ^0}$, we have, for $1\leq j \leq j(\theta ^0),\ CE_j(\theta )\geq \varpi_{\theta ^0} (\theta)$ and, for $j\geq j(\theta ^0),\ CE_j(\theta) \geq 2d$. We remark that the neighborhood $\widetilde{V}_{\theta ^0}$ can be chosen such that $V_{\theta ^0}\subset \widetilde{V}_{\theta^0}$, where $V_{\theta^0}$ was defined in Lemma \ref{lemMM2}.

\noindent Let $H^b_{\theta^0}(\theta)$ be the $n_{0} \times n_{0} $ diagonal matrix with identical diagonal entries equal to $\varpi_{\theta^0}$. For $\theta \in \widetilde{V}_{\theta ^0}$, we have
\begin{equation}
H ^b _{\theta ^0}(\theta)\leq C\cdot H^0 (\theta). \label{24}
\end{equation}
Finally, we note that $(\widetilde{V}_{\theta ^0})_{\theta ^0 \in Z}$ can be chosen so that they cover $\T^*$, (i.e. $\displaystyle \cup _{\theta ^0 \in Z} \widetilde{V}_{\theta^0}=\T^*$) and such that each one of them contains only one point of $Z$ (i.e. for $\theta\in Z, \theta^{\prime}\in Z $ such that $\theta \neq \theta ^{\prime}$, we have $\theta^{\prime} \notin \overline{\widetilde{V}_{\theta}}$). We order the points in $Z=\{\theta^k; 1\leq k\leq m_0 \}$, where $m_0 = \# Z$. Let $(\chi_k)_{1\leq k\leq m_0}$ be functions in $C^{\infty}(\T^*)$ which form a partition of the unity on $\T^*$ such that, for $1\leq k \leq m_0$, $\supp (\chi_k) \subset \widetilde{V}_{\theta^k}$, $0\leq \chi_k \leq 1 $ and $\chi_k \equiv 1$ in a neighborhood of $\theta^k$.

\noindent So there exists $C>1$ such that, for any $\theta \in \T^*$, we have,
\begin{equation}\label{vv}
\frac{1}{m_0} \leq \sum _{k=1}^{m_0} \chi_k^2 \leq 1\quad \mbox{ and }\quad \sum_{k=1}^{m_0} H^b_{\theta^k}(\theta)\chi_k ^2 \leq C H^0(\theta).
\end{equation}
For $t\in (L^2(\T^*)\otimes \C^{\DD})\otimes \C^{n_{0}}\otimes \C^{m_0},$ we set $t=(t_{j,k})_{1\leq j\leq n_{0};1\leq k\leq m_0}$. We consider $t$ as a system of $m_0$ columns denoted by $(t_{.,k})_{1\leq k\leq m_0}$. Each column belongs to $(L^2(\T^*)\otimes \C^{\DD})\otimes \C^{n_{0}}$. We endow $(L^2(\T^*)\otimes \C^{\DD})\otimes \C^{n_{0}}\otimes \C^{m_0}$ with the scalar product generating the following Euclidean norm:
$$\Big\| t\Big\|_{(L^2(\T^*)\otimes \C^{\DD})\otimes \C^{n_{0}}\otimes \C^{m_0}}^2=\sum_{k=1}^{m_0} \Big\| t_{\cdot,k}\Big\| _{(L^2(\T^*)\otimes \C^{\DD})\otimes \C^{n_{0}}}^2=\sum_{1\leq j\leq n_{0}
,1\leq k\leq m_0}\Big\| t_{j,k}\Big\| _{L^2(\T^*)\otimes \C^{\DD}}^2.$$
We define the mapping $S:(L^2(\T^*)\otimes \C^{\DD})\otimes \C^{n_{0}}\longrightarrow (L^2(\T^*)\otimes \C^{\DD})\otimes \C^{n_{0}}\otimes \C^{m_0} $ by
$$S(t)=(\chi_k t)_{1\leq k\leq m_0}=(\chi_k t_j)_{1\leq j\leq n_{0},\ 1\leq k\leq m_0},\quad \mbox{ if }\ t=(t_j)_{1\leq j\leq n_{0}}\in (L^2(\T^*)\otimes \C^{\DD})\otimes \C^{n_{0}}.$$
Here, for any $1\leq j \leq n_0$, $t_j=(t_{ij})_{1\leq i\leq D}\in L^{2}(\T^*)\otimes \C^{\DD}$.

\noindent The adjoint of $S$, $S^{*}:(L^2(\T^*)\otimes \C^{\DD})\otimes \C^{n_{0}}\otimes \C^{m_0}\longrightarrow (L^2(\T^*)\otimes \C^{\DD})\otimes \C^{n_{0}}$ is defined by
$$S^{*}(t)=\Bigg( \sum_{1\leq k\leq m_0}\chi_k t_{j,k}\Bigg) _{1\leq j\leq n_{0}}\quad \mbox{for}\ t=(t_{j,k})_{1\leq j\leq n_{0}; 1\leq k\leq m_0} \in (L^2(\T^*)\otimes \C^{\DD})\otimes \C^{n_{0}}\otimes \C^{m_0}.$$
Here, for any $1\leq j\leq n_0$ and any $1 \leq k \leq m_0$, we have $t_{j,k}=(t_{i,j,k})_{1\leq i\leq D}$. According to equation (\ref{vv}) we have $\frac{1}{m_0} I\leq S^{*}\circ S\leq I$, (here $I$ is the identity in $(L^2(\T^*)\otimes \C^{\DD})\otimes \C^{n_{0}}$), thus $S$ is one to one.  Using the boundedness assumption on the $V_i$ and on the support of the $\omega_i^{(n)}$, one shows the following lemma:

\begin{lem}\label{6.5}
There exists $C>0$ such that, for $t\in (L^2(\T^*)\otimes \C^{\DD})\otimes \C^{n_{0}}$, we have
$$\langle H_\omega^a S(t),S(t) \rangle_{(L^2(\T^*)\otimes \C^{\DD})\otimes \C^{n_{0}}\otimes \C^{m_0}}\leq C\langle H_\omega^{0}t,t\rangle_{(L^2(\T^*)\otimes \C^{\DD})\otimes \C^{n_{0}}},$$
where the operator $H_{\omega}^a$ acting on $(L^2(\T^*)\otimes \C^{\DD})\otimes \C^{n_{0}}\otimes \C^{m_0}$ is defined by
$$H_\omega^a t=\Big( H_k^a t_{i,j,k}+ V_{\omega, i}^a t_{i,j,k}\Big)_{1\leq i \leq D,\ 1\leq j\leq n_{0},\ 1\leq k\leq m_0}.$$
Here, $H_k^a$ is the multiplication by $\varpi_{\theta^k}$ acting as
a multiplication operator on $L^2(\T^*)$, $V_{\omega,i}^a=\sum_{n
\in \Z^d} \omega_i^{(n)}\Pi_{n}$, where $\Pi_{n}$ is the orthogonal
projection on the vector $\theta\mapsto e^{\ii n \theta}$ in
$L^2(\T^*)$, and $H_{\omega}^0$ is defined in Section
\ref{sec_disc_pb}. For $A=(a_{i,j})_{1\leq i,j\leq n_0},$ and $t\in
(L^2(\T^*)\otimes \C^{\DD})\otimes \C^{n_{0}}, At\in
(L^2(\T^*)\otimes \C^{\DD})\otimes \C^{n_{0}}$, with
$(At)_{j}=\sum_{i=1}^{n_0}a_{j,i}t_i$ \end{lem}

\noindent The proof of this lemma follow the same steps as Lemma 5.5 in \cite{N03}. We use it to end the proof of Theorem \ref{LL12}. Let us first notice that the operator $H_{\omega}^a$ could be written as a direct sum of $n_0$ copies of $m_0\times \DD$ random scalar-valued continuous Anderson models. Indeed, we can write
$$(L^2(\T^*)\otimes \C^{\DD})\otimes \C^{n_0}\otimes \C^{m_0}=\bigoplus_{1\leq i\leq D,\ 1\leq j\leq n_0,\ 1\leq
k\leq m_0} L^{2}(\T^*)\otimes \tilde{\C^{i}}\otimes \tilde{\C^{j}}\otimes {\tilde{\C^{k}}}.$$
Here, for $1\leq j \leq l$, we use the notation $\tilde{\C^{j}}=\{0\}^{j-1}\times \C\times \{0\}^{l-{j}}$. So $H_\omega^a$ is unitarily equivalent to
$$\bigoplus_{1\leq i \leq D,\ 1\leq j\leq n_0,\ 1\leq k\leq m_0}\ H_{\omega,i,k}^{\mathrm{And}},$$
Here $H_{\omega,i,k}^{\mathrm{And}}$ acts on  $L^{2}(\T^*)\otimes \tilde{\C^{i}}\otimes \tilde{\C^{j}}\otimes{\tilde{\C^{k}}}$. Using the discrete Fourier transformation, we get that for every $k\in \{ 1,\ldots, m_0 \}$,  $H_{\omega,i,k}^{\mathrm{And}}$ is unitarily equivalent to $ h_{\omega,i}^{\mathrm{And}}$, where $h_{\omega,i}^{\mathrm{And}}$ acts on $\ell^2(\Z^d)$ and is defined by
\begin{equation}\label{rafik}
h_{\omega,i}^{\mathrm{And}} = -\Delta_{\Z^d}+\sum_{n \in \Z^d} \omega_i^{(n)}\pi_n.
\end{equation}
Here, if $\delta_n$ is the vector $(\delta_{m}^n)_{\beta \in \Z^d}$ where $\delta_{m}^n$ is the Kronecker's symbol, then $\pi_n$ is the orthogonal projection on $\delta_n$  and $-\Delta_{\Z^d}$ is the discrete Laplacian defined by :
\begin{equation}
\forall u\in l^2(\Z^d),\ (\Delta_{\Z^d} u)_n=\frac{1}{2}\sum_{|m-n|=1}(u_n - u_m).
\end{equation}
Using the fact that for operators $A$ and $B$, we have $N(A\oplus B,E)=N(A,E)+N(B,E)$ (see \cite{PF}), we get that
\begin{equation}\label{rafik2}
N_{\mathcal{E}_0}(\varepsilon)\leq n_0\times m_0 \times \sum_{i=1}^{\DD} N( h_{\omega,i}^{\mathrm{And}},C. m_0 .\varepsilon).
\end{equation}

\noindent To satisfy assumptions of \cite{S85}, we set, for every $i\in \{1,\ldots, \DD\}$, $s_i=\sup_{n\in \Z^d} \omega_i^{(n)}$ and :

$$\forall i\in \{ 1,\ldots,\DD\},\ \tilde{\omega}_i^{(n)}= \left\lbrace \begin{array}{ccc}
                                        0 & \mbox{ if } & \omega_i^{(n)}\in [0,s_i /2] \\
                    s_i /2 & \mbox{ if } &  \omega_i^{(n)}\in (s_i /2,s_i ]
                                       \end{array}\right.$$

\vskip 2mm

\noindent By changing $\omega_i^{(n)}$ into $\tilde{\omega}_i^{(n)}$ in (\ref{rafik}), we define a new operator which we denote by $\tilde{h}_{\omega,i}^{\mathrm{And}}$. We notice that $\tilde{h}_{\omega,i}^{\mathrm{And}}$ lower bound $h_{\omega,i}^{\mathrm{And}}$ with the same bottom of the spectrum. As it is known that each $\tilde{h}_{\omega,i}^{\mathrm{And}}$ exhibits Lifshitz tails with Lifshitz exponent $-d/2$ (see \cite{K89,S85}), using Theorem \ref{thm_bb22} and (\ref{rafik2}), we get that
$$\limsup_{\varepsilon \longrightarrow 0^{+}}\frac{\log | \log (N(\varepsilon )-N(0^+))|}{\log \varepsilon} \leq -\frac{d}{2}.$$
This ends the proof of Theorem \ref{LL12}.$\Box$

\vskip 5mm

\acknowledgements Both authors would like to thank Fr\'ed\'eric
Klopp for fruitful discussions. H. Boumaza would also like to thank
H.N, for its hospitality on two occasions.


\end{document}